\documentclass[11pt]{article}
\usepackage{amsmath,amssymb,amsthm,easybmat,verbatim,float,hyperref}
\usepackage{enumerate}
\usepackage{color}
\usepackage{bm}
\usepackage{graphicx}
\usepackage{longtable}
\usepackage{array}
\newcolumntype{x}[1]{%
>{\centering\arraybackslash}p{#1}}%

\newcommand{\sepAuthor}{0.5in}
\newcommand{\sepAbstract}{0.4in}
\newcommand{\skipKeywords}{30pt}

\long\def\mytitlepage#1#2#3#4{
        \thispagestyle{empty}
        \begin{center}
        {\Large\bf #1}

        \vspace{\sepAuthor}
        #2\\
        \medskip

        \vspace{\sepAbstract}
        {\Large Abstract}
        \end{center}

        \noindent{#3}
        \vskip\skipKeywords

        \noindent{#4}
        \clearpage
        }

\newtheorem{remark}{Remark}

\newtheorem{proposition}{Proposition}

\newcommand{\ket}[1]{|#1\rangle}

\begin{document}
\mytitlepage{Comment on ``Matrix Pencils and Entanglement Classification''\!
\footnote{This work was supported in part by the National Science
Foundation of the United States under Awards~0347078 and 0622033.
}}{
Eric Chitambar\\
Physics Department, University of Michigan\\
450 Church Street, Ann Arbor, Michigan 48109-1040, USA\\
E-mail: echitamb@umich.edu\\
\vspace{2ex}
Carl A. Miller\\
Department of Mathematics, University of Michigan\\
530 Church Street,  
Ann Arbor, MI 48109-1043, USA\\
E-mail: carlmi@umich.edu\\
\vspace{2ex}
Yaoyun Shi\\
Department of Electrical Engineering and Computer Science\\
University of Michigan\\ 
2260 Hayward Street, Ann Arbor, MI 48109-2121, USA\\
E-mail: shiyy@umich.edu
}{\noindent In our earlier posting ``Matrix Pencils and Entanglement Classification'', \href{http://lanl.arxiv.org/abs/0911.1803}{arXiv:0911.1803}, we gave a polynomial-time algorithm for deciding if two states in a space of dimension $2\otimes m\otimes n$ 
are SLOCC equivalent. In this note, we point out that a straightforward modification of the algorithm gives
a simple enumeration of all SLOCC equivalence classes
in the same space, with the class representatives expressed in the Kronecker canonical normal form of matrix pencils.  Thus, two states are equivalent if and only if they have the same canonical form.  As an example, we present representatives in canonical form for each of the 26 equivalence classes in $2\otimes 3\otimes n$ systems.
}{}

We assume that the reader is familiar with our recent posting \href{http://lanl.arxiv.org/abs/0911.1803}{arXiv:0911.1803}.
We will define a canonical form of a state in dimensions $2\otimes m\otimes n$ which we call the \textbf{State Kronecker Canonical Form (SKCF)}.  An approach also based on analyzing pairs of matrices was taken by 
researchers in Ref. \cite{Cheng-2009a}.  There, the authors consider exclusively $2\otimes n\otimes n$ 
systems but unfortunately err in deriving Theorems 1 and 2 and thus miss a 
whole range of equivalence classes.  Here and in our earlier posting, we correct this mistake primarily by identifying linear fractional transformations (LFTs) as an essential ingredient in SLOCC transformations.  The following construction of the SKCF relies heavily on the following fact.

\begin{proposition}\label{prop:lft}
Let $x=(x_1, x_2, x_3)$ and $y=(y_1, y_2, y_3)$ be two triples of distinct values in $\mathbb{C}^*=\mathcal{C}\cup\{\infty\}$.
Then there exists a unique linear fractional transformation $\ell=\ell(a, b, c, d): z\mapsto \frac{az+b}{cz+d}$, where $ac-bd\ne0$,
such that $\ell(x_i)=y_i$, $1\le i\le 3$.
\end{proposition}
Recall that $\ell(\infty)=a/c$, and $\ell(-d/c)=\infty$. Also recall that LFTs form a group under function composition; in particular
each LFT is reversible and its reverse is also a LFT.

Let $\ket{\Psi}\in2\otimes m\otimes n$. We write
\[ \ket{\Psi}=\ket{0}\otimes(\sum_{i=0, j=0}^{m-1,n-1}\alpha_{ij}\ket{ij})+\ket{1}\otimes(\sum_{i=0, j=0}^{m-1,n-1}\beta_{ij}\ket{ij}).\]
With \[ R=[\alpha_{ij}]_{0\le i\le m-1, 0\le j\le n-1}\quad\textrm{and}\quad
S=[\beta_{ij}]_{0\le i\le m-1, 0\le j\le n-1},\]
the corresponding matrix pencil for $\ket{\Psi}$ is
\[ U=U_\Psi=\mu R + \lambda S.\]
To compute the canonical form $F=F_\Psi$, we first compute the Kronecker canonical form (KCF) of $U$, which
is the direct sum of a set of blocks of the following types. We will use the notation in~\cite{CMS09}, in particular those in Lemma~(1) there.
\begin{enumerate}[Type $1$.]
\item[Type $0$.] A zero matrix $F^0$ of dimension $h\times g$.
\item $L_{\epsilon_1}$, $L_{\epsilon_2}$, ...., $L_{\epsilon_u}$, for some integers $u\ge 0$, and $\epsilon_1,..., \epsilon_u$ with
$0<\epsilon_1\le \epsilon_2\le\cdots\le\epsilon_u$.
\item $L^T_{\nu_1}$, $L^T_{\nu_2}$, ..., $L^T_{\nu_v}$, for some integers $v\ge0$ and $\nu_1, ..., \nu_v$ with $0<\nu_1\le \nu_2\le\cdots\le\nu_v$.
\item A set of blocks of regular pencils, determined by a 
sequence of {\em distinct} eigenvalues $x=(x_1, x_2, ..., x_r)$, $x_1, ..., x_r\in\mathbb{C}^*$, and 
a corresponding sequence $\eta=(\eta^1, \eta^2, ..., \eta^r)$, where
$\eta^i=(\eta^i_1, ..., \eta^i_{r_i})$
is itself a sequence of integers with $0<\eta^i_1\le \eta^i_2\le\cdots\le \eta^i_{r_i}$. We call $\eta^i$ the {\em size signature}
of the eigenvalue $x_i$. Each pair of $(i, j)$, $1\le i\le r$, $1\le j\le r_i$,
determines a regular pencil $R(x_i, \eta^i_j)$, which is either
$M^{\eta^i_j}$ (but with $x_i$ replaced by $-x_i$) when $x_i\ne \infty$ or $N^{\eta^i_j}$ otherwise.
\end{enumerate}
In $F$, we will first arrange the above blocks according to their types (from Type 0 to Type 3).
For blocks of Type 1 and Type 2, we order non-decreasingly in size. 
For blocks of Type 3, we first order the size signatures according to the following ordering.
\begin{itemize}
\item First order, non-decreasingly, according to the multiplicities of the size signatures, i.e., the fewer times
a size signature appears, the earlier it appears in the ordering.
\item For size signatures of the same multiplicity, we order them according to a
fixed total ordering of size signatures, such as the `graded lexicographical order', in which
$(\eta_1, ..., \eta_a)\prec (\eta'_1,...., \eta'_{b})$ if $a < b$ or when $a=b$, the first non-zero elements in $\eta_i-\eta_i'$ is negative. 
\end{itemize}
We will assume from now on that $\eta$ is ordered, and that there are $k$ distinct size signatures,
corresponding to $\zeta_1$, ..., $\zeta_k$ number of distinct eigenvalues, respectively. 
Fix a total ordering of $\mathbb{C}^*$. For a sequence of distinct complex numbers 
$y=(y^1_1, y^1_2, ..., y^1_{\zeta_1}, y^2_1,...,y^2_{\zeta_2},...,y^k_{1},..., y^k_{\zeta_k})$,
define the {\em $\eta$-ordered sequence of $y$} as the unique sequence 
$\omega(y)=(z^1_1, z^1_2, ..., z^1_{\zeta_1}, z^2_1,...,z^2_{\zeta_2},...,z^k_{1},..., z^k_{\zeta_k})$, 
where $z^i_1\prec z^i_2\prec\cdots\prec z^i_{\zeta_i}$ is a permutation of $y^i_1, ..., y^i_{\zeta_i}$.
Two sequences of distinct eigenvalues from $x$, $y=(y_1, ..., y_{l})$ and $z=(z_1, ..., z_{l})$
are said to {\em of the same type} if they are of
the same length and the size signatures for $y_i$ and $z_i$ are the same, for all $i$, $1\le i\le l$.
Note that $\omega(y)$ and $y$ are of the same type for any $y$.

We will describe below how to obtain a sequence $\hat x$ of $r$ distinct eigenvalues, not necessarily the same as those
in $x$, to replace $x$ 
in the output canonical form. If $r\le 3$, set $\hat x$ to be the first $r$ elements of $(0, 1, \infty)$. 
If $r\ge 4$, denote by $\bar X$ the set of triples $\bar x=(x_{i_1}, x_{i_2}, x_{i_3})$ of distinct eigenvalues from $x$
that are of the same type as $(x_1, x_2, x_3)$.
Each such triple determines a unique linear fractional transformation $\theta_{\bar x}$ that maps $\bar x$ to $(0, 1, \infty)$.
Denote by $\hat X=\{\omega(\theta_{\bar x}(x)) : \bar x\in \bar X\}$ and finally,
\[ \hat x=\min \hat X\]
where the minimization is over some fixed ordering on sequences of complex numbers.
Output $(\hat x, \eta)$ as the canonical form for the regular pencil blocks in
the SKCF of $\ket{\Psi}$. This completes the description of the SKCF for $\ket{\Psi}$.

\begin{proposition}\label{prop:stateKCF}
Two states are SLOCC equivalent if and only if they have the same state Kronecker canonical form.
\end{proposition}

\begin{proof}
The ``if'' the direction is obvious (though Proposition~\ref{prop:lft} is critical when there is at most 3
distinct eigenvalues). 
Now consider two SLOCC equivalent states $\ket{\Psi}$ and $\ket{\Psi'}$,
whose SKCFs are $F$ and $F'$, respectively. Then $F$ and $F'$ must have the same Type 0, 1, 2, blocks, as well
as the same sequence of size signatures $\eta=(\eta^1, \cdots, \eta^r)$, which is ordered. 
Suppose $x$ and $x'$ are the $\eta$-ordered
eigenvalue sequences of $U_\Psi$ and $U_{\Psi'}$, respectively. We must have for some LFT $\theta_0$,
$\theta_0(x')$ is $x$ with eigenvalues of the same size signatures possibly permuted. Thus we have,
\begin{equation}\label{eqn:theta_0}
\omega(\theta_0(x'))=x,
\end{equation}
and in particular, $(x_1, x_2, x_3)$ and $(x'_1, x'_2, x'_3)$ are of the same type.
Also, for any LFT $\theta$,
\begin{equation}\label{eqn:omega}
\omega(\theta\theta_0 (x')) = \omega(\theta(x)).
\end{equation}

We want to show that $\hat x=\hat x'$. 
This holds trivially when $r\le 3$. Suppose now that $r\ge 4$. Let $\Theta=\{\theta_{\bar x}: \bar x\in\bar X\}$
and $\Theta'=\{\theta_{\bar x'}: \bar x'\in\bar X'\}$. We claim that $\Theta'=\Theta\theta_0$.
To see this, fix an $\theta_{\bar x}\in\Theta$ with $\bar x=(x_{i_1}, x_{i_2}, x_{i_3})$. Then
$\theta_0^{-1}(\bar x)=(x'_{j_1}, x'_{j_2}, x'_{j_3})$ are of the same type as $\bar x$, thus of the same
type of $(x_1, x_2, x_3)$, and again the same type of $(x'_1, x'_2, x'_3)$.  Since $\theta_{\bar x}\theta_0(x'_{j_1},
x'_{j_2}, x'_{j_3})=(0, 1, \infty)$, $\theta_{\bar x}\theta_0\in \bar X'$. Thus $\Theta\theta_0\subseteq\Theta'$.
 Similarly, for each $\theta_{\bar x'}\in\Theta'$,
$\theta_{\bar x'}=\theta_{\bar x}\theta_0$, where $\bar x=\theta_0(\bar x')$. Thus $\Theta'\subseteq\Theta\theta_0$.
Consequently, $\Theta'=\Theta\theta_0$ and $\hat X'=\omega(\Theta' x') = \omega(\Theta\theta_0 x') =\omega(\Theta x)=\hat X$.
Thus $\hat x'=\hat x$, and $F=F'$.
\end{proof}

\begin{remark} The restriction of $\bar x$ having the same type as $(x_1, x_2, x_3)$ is not necessary in defining a canonical
form. But without the restriction (i.e.
$\bar x$ can take any triple of distinct eigenvalues from $x$) the computational cost would be higher for some states
as $r(r-1)(r-3)$ LFTs would have to be considered. The worst case complexities, though, are the same (when all size signatures
are the same).
\end{remark}
\begin{remark} The algorithm in \cite{CMS09} can be modified to be the following:first computing the SKCFs of the two input states,
then checking if they are identical. The worst cast complexity remains the same.
\end{remark}

As examples of the SKCF, we now examine the canonical forms of all 26 classes in $2\otimes 3\otimes n$ systems.  This is the largest tripartite dimensions having a finite number of SLOCC equivalence classes.  We also direct the reader to the work of Chen \textit{et al.} \cite{Chen-2006a} for a different derivation of other representatives for the following classes.

\bigskip\noindent
{\bf $2\otimes 2\otimes 2$ Systems}

Here the states are represented as $2\times 2$ pencils.  We first consider the case with no minimal indices.  Here there can only be two or one distinct elementary divisors with the latter having possible signatures of $\{1,1\}$ and $\{2\}$.  In matrix and bra-ket form, these correspond to the unnormalized states
\begin{align*}
&\bigl(\begin{smallmatrix} \lambda&\cdot\\\cdot&\mu+\lambda\end{smallmatrix}\bigr)\text{\scriptsize(ABC-1) ``GHZ-class''}
&&\bigl(\begin{smallmatrix} \lambda&\cdot\\\cdot&\lambda\end{smallmatrix}\bigr)\text{\scriptsize(A:BC-1)} 
&&\bigl(\begin{smallmatrix} \lambda&\mu\\\cdot&\lambda\end{smallmatrix}\bigr)\text{\scriptsize(ABC-2) ``W-class''} \\
&(\ket{0}+\ket{1})\ket{11}+\ket{100} &&\ket{100}+\ket{111}  &&\ket{001}+\ket{100}+\ket{111}.
\end{align*}

The only possible classes included in three qubit systems are those with Bob and Charlie having non-maximal local ranks.  When $h=1,g=0$, the only possibility is $\epsilon_1=1$, while for $h=0,g=1$ it must be $\nu_1=1$.  The case of $h=1,g=1$, there are no non-zero minimal indices.  These three states are given by
\begin{align*}
&\bigl(\begin{smallmatrix}\cdot&\cdot\\\lambda&\mu\end{smallmatrix}\bigr)\text{\scriptsize{(AC:B)}}
&&\bigl(\begin{smallmatrix}\cdot&\lambda\\\cdot&\mu\end{smallmatrix}\bigr)\text{\scriptsize{(AB:C)}}
&&\bigl(\begin{smallmatrix}\cdot&\cdot\\\cdot&\mu\end{smallmatrix}\bigr)\text{\scriptsize{(A:B:C)}}\\
&\ket{011}+\ket{101}&&\ket{011}+\ket{101}&&\hspace{.2cm}\ket{011}.
\end{align*}
We see that (A:B:C) represents the product states while (AC:B) and (AB:C) are the bipartite pure entanglement with respect to the specified partitioning.

\bigskip\noindent
{\bf $2\otimes 2\otimes 3$ Systems}

Since we are only concerned with the states of maximal local ranks for Bob and Charlie, we only consider pencils having $h=g=0$.  The only possible minimal indices are $\epsilon_1=1$ and $\epsilon_1=2$ which correspond to the states
\begin{align*}
&\bigl(\begin{smallmatrix}\lambda&\mu&\cdot\\\cdot&\cdot&\lambda\end{smallmatrix}\bigr)\text{\scriptsize{(ABC-3)}}
&&\bigl(\begin{smallmatrix}\lambda&\mu&\cdot\\\cdot&\lambda&\mu\end{smallmatrix}\bigr)\text{\scriptsize{(ABC-4)}}\\
&\ket{001}+\ket{100}+\ket{112}&\ket{001}&+\ket{012}+\ket{100}+\ket{111}.
\end{align*}
The state (ABC-3) has a single elementary divisor of $\lambda$ while (ABC-4) has none.  The tensor rank of both these states is three.  In fact, an explicit three-term expansion of (ABC-3) is given by $\tfrac{1}{2}\ket{+_{01}}\ket{+_{01}}\ket{+_{12}}+\tfrac{1}{2}\ket{-_{01}}\ket{-_{01}}\ket{-_{12}}+\ket{1}\ket{0}\ket{+_{02}}$ where $\ket{\pm_{ij}}=\ket{i}\pm\ket{j}$.

\bigskip\noindent
{\bf $2\otimes 2\otimes n$ Systems for $n\geq 4$}

As noted in the discussion above, it is enough to consider $2\otimes 2\otimes 4$ systems.  For states with Bob and Charlie having full local ranks, the only possible minimal indices are $\epsilon_1=\epsilon_2=1$ which corresponds to the state
\begin{align*}
&\bigl(\begin{smallmatrix}\lambda&\mu&\cdot&\cdot\\\cdot&\cdot&\lambda&\mu\end{smallmatrix}\bigr)\text{\scriptsize{(ABC-5)}}\\\ket{001}&+\ket{013}+\ket{100}+\ket{112}.
\end{align*}

\bigskip\noindent
{\bf $2\otimes 3\otimes 2$ Systems}

These pencils are simply the transpose of $2\times 3$ pencils and thus contribute two equivalence classes of states with maximal local ranks:
\begin{align*}
&\Bigl(\begin{smallmatrix}\lambda&\cdot\\\mu&\cdot\\\cdot&\lambda\end{smallmatrix}\Bigr)\text{\scriptsize{(ABC-6)}}=\text{\scriptsize{(ABC-3)$^T$}}
&&\Bigl(\begin{smallmatrix}\lambda&\cdot\\\mu&\lambda\\\cdot&\mu\end{smallmatrix}\Bigr)\text{\scriptsize{(ABC-7)}}=\text{\scriptsize{(ABC-4)$^T$}}\\
&\ket{010}+\ket{100}+\ket{121}&\ket{010}&+\ket{021}+\ket{100}+\ket{111}.
\end{align*}

\bigskip\noindent
{\bf $2\otimes 3\otimes 3$ Systems}

Here we have $3\times 3$ pencils and for those having no minimal indices, the possible collections of eigenvalue signatures are
$\{\{1\},\{1\},\{1\}\}$, $\{\{1,1\},\{1\}\}$, $\{\{1,1,1\}\}$, $\{2,1\}$ $\{\{2\},\{1\}\}$, $\{3\}$ and belong to the representative states
\begin{align*}
&\Bigl(\begin{smallmatrix} \lambda&\cdot&\cdot\\\cdot&\mu+\lambda&\cdot\\\cdot&\cdot&\mu\end{smallmatrix}\Bigr)\text{\scriptsize{(ABC-8)}}
&&\Bigl(\begin{smallmatrix} \lambda&\cdot&\cdot\\\cdot&\lambda&\cdot\\\cdot&\cdot&\mu+\lambda\end{smallmatrix}\Bigr)\text{\scriptsize{(ABC-9)}}
&&\Bigl(\begin{smallmatrix} \lambda&\cdot&\cdot\\\cdot&\lambda&\cdot\\\cdot&\cdot&\lambda\end{smallmatrix}\Bigr)\text{\scriptsize{(A:BC-2)}}\\
\ket{100}+&(\ket{0}+\ket{1})\ket{11}+\ket{022}
&\ket{100}&+\ket{111}+\ket{022}
&\ket{100}&+\ket{111}+\ket{122}
\end{align*}
\begin{align*}
&\Bigl(\begin{smallmatrix} \lambda&\mu&\cdot\\\cdot&\lambda&\cdot\\\cdot&\cdot&\lambda\end{smallmatrix}\Bigr)\text{\scriptsize{(ABC-10)}}
&&\Bigl(\begin{smallmatrix} \lambda&\mu&\cdot\\\cdot&\lambda&\cdot\\\cdot&\cdot&\mu+\lambda\end{smallmatrix}\Bigr)\text{\scriptsize{(ABC-11)}}
&&\Bigl(\begin{smallmatrix} \lambda&\mu&\cdot\\\cdot&\lambda&\mu\\\cdot&\cdot&\lambda\end{smallmatrix}\Bigr)\text{\scriptsize{(ABC-12)}}\\
\ket{001}+&\ket{100}+\ket{111}+\ket{122}
&\ket{00&1}+\ket{100}+\ket{111}
&\ket{00&1}+\ket{012}+\ket{100}
\\
&&&+(\ket{0}+\ket{1})\ket{22}
&&+\ket{111}+\ket{122}.
\end{align*}

For $3\times 3$ pencils, the only possible minimal indices are $\epsilon_1=\nu_1=1$ corresponding to the representative state 
\begin{align*}
&\Bigl(\begin{smallmatrix} \lambda&\mu&\cdot\\\cdot&\cdot&\mu\\\cdot&\cdot&\lambda\end{smallmatrix}\Bigr)\text{\scriptsize{(ABC-13)}}\\
\ket{001}+&\ket{012}+\ket{100}+\ket{122}.
\end{align*}

\bigskip\noindent
{\bf $2\otimes 3\otimes 4$ Systems}

For a minimal indice of $\epsilon_1=1$, we have the classes represented by
\begin{align*}
&\Bigl(\begin{smallmatrix}\lambda&\mu&\cdot&\cdot\\\cdot&\cdot&\lambda&\cdot\\\cdot&\cdot&\cdot&\lambda\end{smallmatrix}\Bigr)\text{\scriptsize{(ABC-14)}}
&&\Bigl(\begin{smallmatrix}\lambda&\mu&\cdot&\cdot\\\cdot&\cdot&\lambda&\cdot\\\cdot&\cdot&\cdot&\lambda+\mu\end{smallmatrix}\Bigr)\text{\scriptsize{(ABC-15)}}
&&\Bigl(\begin{smallmatrix}\lambda&\mu&\cdot&\cdot\\\cdot&\cdot&\lambda&\mu\\\cdot&\cdot&\cdot&\lambda\end{smallmatrix}\Bigr)\text{\scriptsize{(ABC-16)}}\\
\ket{001}+\ket{10&0}+\ket{112}+\ket{123}
&\ket{001&}+\ket{100}+\ket{112}
&\ket{001&}+\ket{013}+\ket{100}\\
&&+(&\ket{0}+\ket{1})\ket{23}
&+\ket{1&12}+\ket{123}.
\end{align*}
We also have the states with $\epsilon_1=2$ and $\epsilon_1=3$ respectively:
\begin{align*}
&\Bigl(\begin{smallmatrix}\lambda&\mu&\cdot&\cdot\\\cdot&\lambda&\mu&\cdot\\\cdot&\cdot&\cdot&\lambda\end{smallmatrix}\Bigr)\text{\scriptsize{(ABC-17)}}
&&\Bigl(\begin{smallmatrix}\lambda&\mu&\cdot&\cdot\\\cdot&\lambda&\mu&\cdot\\\cdot&\cdot&\lambda&\mu\end{smallmatrix}\Bigr)\text{\scriptsize{(ABC-18)}}
\\
\ket{001}+\ket{012}&+\ket{100}+\ket{111}+\ket{123}
&\ket{001}+\ket{012}+\ket{&023}+\ket{100}+\ket{111}+\ket{122}.
\end{align*}

\bigskip\noindent
{\bf $2\otimes 3\otimes 5$ Systems}

The possibilities are $\epsilon_1=1,\epsilon_2=1$ and $\epsilon_1=1,\epsilon_2=2$ corresponding to
\begin{align*}
&\Bigl(\begin{smallmatrix}\lambda&\mu&\cdot&\cdot&\cdot\\\cdot&\cdot&\lambda&\mu&\cdot\\\cdot&\cdot&\cdot&\cdot&\lambda\end{smallmatrix}\Bigr)\text{\scriptsize{(ABC-19)}}
&&\Bigl(\begin{smallmatrix}\lambda&\mu&\cdot&\cdot&\cdot\\\cdot&\cdot&\lambda&\mu&\cdot\\\cdot&\cdot&\cdot&\lambda&\mu\end{smallmatrix}\Bigr)\text{\scriptsize{(ABC-20)}}
\\
&\ket{001}+\ket{013}+\ket{100}
&&\ket{001}+\ket{013}+\ket{024}
\\
&+\ket{112}+\ket{124}
&&+\ket{100}+\ket{112}+\ket{123}.
\end{align*}
 
\bigskip\noindent
{\bf $2\otimes 3\otimes n$ Systems for $n\geq 6$}

We must only consider $n=6$ which allows for $\epsilon_1=1,\epsilon_2=1,\epsilon_3=1$ with representative
\begin{align*}
&\Bigl(\begin{smallmatrix} \lambda&\mu&\cdot&\cdot&\cdot\\\cdot&\cdot&\lambda&\mu&\cdot&\cdot\\\cdot&\cdot&\cdot&\cdot&\lambda&\mu\end{smallmatrix}\Bigr)\text{\scriptsize{(ABC-21)}}\\
\ket{001}+\ket{013}&+\ket{025}+\ket{100}+\ket{112}+\ket{124}.
\end{align*}

\bibliographystyle{abbrv}
\bibliography{QuantumBib}

\end{document}